\pdfoutput=1
\documentclass[11pt,a4paper]{article}

\usepackage[utf8]{inputenc}
\usepackage[T1]{fontenc}
\usepackage{lmodern}
\usepackage{amsmath}
\usepackage{amssymb}
\usepackage{amsthm}
\usepackage{graphicx}
\usepackage{subcaption}
\usepackage{cite}
\usepackage{microtype}
\usepackage[hidelinks]{hyperref}

\hypersetup{
    pdftitle={Small Denominators and Subresonant Accumulation in Weakly Nonlinear Dispersive Dynamics},
    pdfauthor={P.Yu. Astafieva and O.M. Kiselev},
    pdfkeywords={small denominators, subresonance, weakly nonlinear dispersive equations, Klein--Gordon equation, long-time asymptotics}
}

\newtheorem{theorem}{Theorem}
\newtheorem{lemma}{Lemma}
\newtheorem{proposition}{Proposition}

\newtheorem{definition}{Definition}

\numberwithin{equation}{section}

\title{Small Denominators and Subresonant Accumulation in Weakly Nonlinear Dispersive Dynamics}
\author{
P.Yu. Astafieva\\
Russian State Social University
\and
O.M. Kiselev\\
Innopolis University
}
\date{July 2, 2026}

\begin{document}

\maketitle

\begin{abstract}
We study a small-denominator mechanism in weakly nonlinear dispersive dynamics.
After Fourier decomposition, a nonlinear dispersive equation becomes an
infinite system of weakly coupled oscillators.  Higher-order correction terms
may then contain infinite families of nonresonant Fourier interactions whose
detunings tend to zero.  Such families do not produce exact secular terms, but
their accumulated contribution may grow as a power of time.  We call this
effect subresonant accumulation.  The rigorous part of the paper is the
analysis of a model forced oscillator and of an abstract subresonant Duhamel
sum.  If the detuning and coefficients have the form
$\Delta_n\sim c n^{-p}$ and $B_n\sim b n^{-\kappa}$, then the accumulated
contribution grows as $t^{1-\alpha}$, where
$\alpha=(\kappa-1)/p$.  We then show how this mechanism appears in a quartic
Fourier family for the Klein--Gordon dispersion law.  For the full nonlinear
partial differential equation we formulate a conditional approximation result:
provided that all remaining resonant and almost resonant interactions are
controlled, the subresonant term gives the leading long-time correction.
\end{abstract}



\section{Introduction}

Small denominators are a classical obstruction in long-time perturbation
theory, from the works of Lindstedt and Poincare to the modern theory of
small divisors \cite{Lindstedt1883,Poincare1899,Arnold1963}.  They appear
when a forcing frequency coincides, or nearly coincides, with an eigenfrequency
of the linearized problem.  Exact resonances generate secular terms, while
uniformly nonresonant terms are usually bounded.  There is, however, an
intermediate possibility: an infinite family of nonresonant terms may have
detunings tending to zero.  We call such families \emph{subresonant}.

The purpose of this paper is to isolate this intermediate regime in weakly
nonlinear dispersive waves \cite{Whitham1974}.  The key mechanism is already
visible in a scalar forced oscillator.  If the forcing contains frequencies
$1-n^{-p}$ with coefficients $n^{-k}$, then the response grows like
$t^{1-\alpha}$, where $\alpha=(k-1)/p$.  This is slower than the linear
secular growth of an exact resonance, but it is still unbounded for
$0<\alpha<1$.

The same mechanism arises in partial differential equations after Fourier
decomposition.  Each spatial Fourier mode of the first nonlinear correction
satisfies a forced oscillator equation.  The forcing frequencies are finite
signed sums of the linear dispersion law.  Higher nonlinear correction terms
may therefore generate infinite families of such sums approaching a linear
frequency.  A full justification for the nonlinear equation requires control of
all resonant and almost resonant interactions.  In this paper this control is
kept as an explicit hypothesis, while the subresonant mechanism itself is
proved.

It is important that the infinitude of a subresonant family is not tied only to
the infinite Taylor series of an analytic nonlinearity.  It already appears
for a fixed polynomial term because a partial differential equation has
infinitely many spatial Fourier modes.  Consequently, the usual separation
between exact resonances and bounded nonresonant terms is not sufficient:
nonresonant Fourier interactions may have vanishing detunings, and their total
contribution may become unbounded.

Intermediate secular effects are known in near-resonant dynamics, modified
scattering, and wave turbulence.  However, to the best of our knowledge, the
power-law accumulation of an infinite family of nonresonant Fourier
interactions with vanishing detuning, leading to an explicit law
$t^{1-\alpha}$, has not previously been isolated as a separate mechanism for
weakly nonlinear dispersive partial differential equations.

Long-time control of nonlinear dispersive equations is usually based on normal
forms, KAM methods, or related procedures for eliminating nonresonant terms
\cite{Kuksin1993,CraigWayne1993,BambusiGrebert2006}.  For Klein--Gordon type
equations these ideas underlie normal form transformations and almost global
existence results for small solutions
\cite{Shatah1985,DelortSzeftel2006,BambusiDelortGrebertSzeftel2007}.  Our
approach is different: instead of eliminating an infinite almost resonant
family, we compute its accumulated contribution.  Previous calculations for
linear model oscillators and parametric subresonance were given in
\cite{AstafyevaKiselev2021,AstafyevaKiselev2022}; here the mechanism is
embedded into the first correction for a nonlinear dispersive equation, in the
spirit of the perturbative construction used for the perturbed
Klein--Fock--Gordon equation \cite{Kiselev1987}.

\section{Main Result and Structure}

The main rigorous result concerns the subresonant sum.  Suppose that, in the
right-hand side of the first-correction equation for a fixed Fourier mode,
there is a family of terms with detunings
\begin{equation}
    \Delta_n\sim c n^{-p},\qquad c\ne0,\qquad p>0,
    \label{eq:intro-detuning}
\end{equation}
and coefficients
\begin{equation}
    B_n\sim b n^{-\kappa},\qquad b\ne0 .
    \label{eq:intro-coefficients}
\end{equation}
The contribution of this family has the form
\begin{equation}
    \sum_{n=1}^{\infty}
    B_n\frac{1-\exp(-i\Delta_n t)}{\Delta_n}.
    \label{eq:intro-subseries}
\end{equation}
If
\begin{equation}
    0<\alpha=\frac{\kappa-1}{p}<1,
    \label{eq:intro-alpha}
\end{equation}
then this sum grows as $t^{1-\alpha}$.  If all other interactions in the same
Fourier mode give bounded contributions or smaller-order terms, then the first
correction contains a subresonant part
\begin{equation}
    u_{1,\mathrm{sub}}=O(t^{1-\alpha}),
    \qquad t\to\infty.
    \label{eq:intro-growth}
\end{equation}
Thus the first-order perturbation expansion remains asymptotic on the time
scale
\begin{equation}
    \varepsilon t^{1-\alpha}\ll1,
    \qquad
    t\ll\varepsilon^{-1/(1-\alpha)}.
    \label{eq:intro-time-scale}
\end{equation}

The paper is organized as follows.  First we derive the oscillator equation
for a Fourier mode of the first correction.  Then we construct an explicit
quartic subresonant family for the Klein--Gordon dispersion law.  The
subresonant growth is proved for a model oscillator and for an abstract
Duhamel sum.  Finally, we formulate a conditional approximation theorem for
the original partial differential equation and illustrate the asymptotics
numerically.

\section{Weakly Nonlinear Dispersive Equation}

We consider a weakly nonlinear dispersive equation on the one-dimensional
torus,
\begin{equation}
    \partial_t^2 u + A(D_x)u = \varepsilon f(u),
    \qquad 0<\varepsilon\ll1,
    \label{eq:pde}
\end{equation}
where $A(D_x)$ is a self-adjoint Fourier multiplier with positive symbol
\begin{equation}
    A(k)=\omega^2(k)>0,\qquad k\in\mathbb Z .
\end{equation}
The nonlinearity is analytic near zero,
\begin{equation}
    f(u)=\sum_{r\ge2}f_ru^r .
\end{equation}

We seek a formal expansion
\begin{equation}
    u(x,t,\varepsilon)
    \sim
    \sum_{j=0}^{N}\varepsilon^j u_j(x,t),
    \qquad N\in\mathbb N .
    \label{eq:formal-expansion}
\end{equation}
At leading order,
\begin{equation}
    \partial_t^2u_0+A(D_x)u_0=0 .
\end{equation}
For real-valued solutions we write
\begin{equation}
    u_0(x,t)
    =
    \sum_{k\in\mathbb Z}
    \left(
        a_k^+e^{i(kx+\omega(k)t)}
        +
        a_k^-e^{i(kx-\omega(k)t)}
    \right),
    \label{eq:u0}
\end{equation}
with the usual complex conjugacy conditions.

\section{First Correction and Small Denominators}

The first correction is governed by
\begin{equation}
    \partial_t^2u_1+A(D_x)u_1=f(u_0).
    \label{eq:first-correction}
\end{equation}
Write
\begin{equation}
    u_1(x,t)=\sum_{q\in\mathbb Z}\widetilde u_1(q,t)e^{iqx}.
\end{equation}
Then each Fourier mode satisfies the forced oscillator equation
\begin{equation}
    \frac{d^2}{dt^2}\widetilde u_1(q,t)
    +
    \omega^2(q)\widetilde u_1(q,t)
    =
    F_q(t).
    \label{eq:mode-oscillator}
\end{equation}

Expanding $f(u_0)$ by \eqref{eq:u0}, the right-hand side has the form
\begin{equation}
    F_q(t)
    =
    \sum_{r\ge2}
    f_r
    \sum_{\sigma\in\{\pm1\}^{r}}
    \sum_{l_1+\cdots+l_r=q}
    A_{l,\sigma}
    \exp\left(
        it\sum_{j=1}^{r}\sigma_j\omega(l_j)
    \right),
    \label{eq:forcing-mode}
\end{equation}
where
\begin{equation}
    A_{l,\sigma}
    =
    \prod_{j=1}^{r}a_{l_j}^{\sigma_j}.
\end{equation}
Thus the forcing frequencies are signed sums
\begin{equation}
    \Omega_\lambda=
    \sum_{j=1}^{r}\sigma_j\omega(l_j),
    \qquad
    \lambda=(r,\sigma,l_1,\ldots,l_r).
\end{equation}

\begin{definition}
For a fixed spatial mode $q$, the detuning of an interaction $\lambda$ is
\begin{equation}
    \Delta_\lambda=\omega(q)-\Omega_\lambda .
\end{equation}
The interaction is resonant if $\Delta_\lambda=0$.  A sequence of nonresonant
interactions $\{\lambda_n\}$ is called subresonant if
\begin{equation}
    \Delta_{\lambda_n}\to0,
    \qquad n\to\infty .
\end{equation}
\end{definition}

\section{Quartic Subresonances for the Klein--Gordon Dispersion}

As a concrete example, consider the nonlinear Klein--Gordon equation
\begin{equation}
    \partial_t^2u-\partial_x^2u+u=\varepsilon f(u),
    \label{eq:kg}
\end{equation}
with dispersion law
\begin{equation}
    \omega(k)=\sqrt{k^2+1}.
\end{equation}
For a quartic term $f(u)=f_4u^4$, the first correction contains frequencies
\begin{equation}
    \Omega
    =
    \sigma_1\omega(l_1)+\sigma_2\omega(l_2)
    +
    \sigma_3\omega(l_3)+\sigma_4\omega(l_4),
\end{equation}
where
\begin{equation}
    l_1+l_2+l_3+l_4=q.
\end{equation}

Take $q=1$, $l_1=1$, and
\begin{equation}
    l_2=n,\qquad l_3=n+1,\qquad l_4=-2n-1 .
\end{equation}
Then $l_1+l_2+l_3+l_4=1$.  Choose the sign combination
\begin{equation}
    \Omega_n=
    \omega(1)-\omega(n)-\omega(n+1)+\omega(2n+1).
    \label{eq:kg-omega-family}
\end{equation}
Since $\Delta_n=\omega(1)-\Omega_n$, we obtain
\begin{equation}
    \Delta_n=
    \omega(n)+\omega(n+1)-\omega(2n+1).
    \label{eq:kg-detuning-family}
\end{equation}
As $n\to\infty$,
\begin{equation}
    \sqrt{n^2+1}
    +
    \sqrt{(n+1)^2+1}
    -
    \sqrt{(2n+1)^2+1}
    =
    \frac{3}{4n}
    -
    \frac{3}{8n^2}
    +
    O(n^{-3}).
    \label{eq:quartic-detuning}
\end{equation}
Thus quartic Fourier interactions generate an explicit subresonant family with
detuning of order $n^{-1}$.

Let the leading Fourier coefficients satisfy
\begin{equation}
    |a_m^\pm|\asymp |m|^{-s},
    \qquad |m|\to\infty .
    \label{eq:fourier-decay}
\end{equation}
The coefficient of this quartic interaction is
\begin{equation}
    B_n=
    C_{\mathrm{comb}}f_4
    a_1^+a_n^-a_{n+1}^-a_{-2n-1}^+ .
    \label{eq:quartic-coefficient}
\end{equation}
Hence
\begin{equation}
    B_n\sim b n^{-3s}
\end{equation}
for an appropriate nonzero constant $b$, provided that the corresponding
asymptotic amplitudes are nonzero.  In the notation above, $p=1$ and
$\kappa=3s$.  Therefore the subresonant power-law regime is
\begin{equation}
    0<\alpha=3s-1<1.
    \label{eq:kg-alpha}
\end{equation}

\section{Model Subresonant Oscillator}

The subresonant mechanism can be isolated in the scalar oscillator
\begin{equation}
    u''+u=f(t),
    \label{eq:model-oscillator}
\end{equation}
with almost periodic forcing
\begin{equation}
    f(t)
    =
    \sum_{n=1}^{\infty}
    \frac{1}{n^k}
    \cos\left(\left(1-\frac{1}{n^p}\right)t\right),
    \qquad p>0.
    \label{eq:model-forcing}
\end{equation}
We impose zero initial conditions,
\begin{equation}
    u(0)=0,\qquad u'(0)=0.
\end{equation}

The solution is
\begin{equation}
    u(t)=\int_0^t\sin(t-s)f(s)\,ds .
\end{equation}
For one frequency $\beta_n=1-n^{-p}$, direct integration gives
\begin{equation}
    \int_0^t\sin(t-s)\cos(\beta_ns)\,ds
    =
    \frac{\cos(\beta_nt)-\cos t}{1-\beta_n^2}.
    \label{eq:single-mode-response}
\end{equation}
Since
\begin{equation}
    1-\beta_n^2=
    \frac{2}{n^p}-\frac{1}{n^{2p}},
\end{equation}
the small denominator is of order $n^{-p}$.

The leading growing terms reduce to the two series
\begin{equation}
    \sigma_s(t)
    =
    \sum_{n=1}^{\infty}
    n^{p-k}\sin\left(\frac{t}{n^p}\right),
    \qquad
    \sigma_c(t)
    =
    \sum_{n=1}^{\infty}
    n^{p-k}\sin^2\left(\frac{t}{2n^p}\right).
    \label{eq:sigma-series}
\end{equation}
Let
\begin{equation}
    \alpha=\frac{k-1}{p}.
\end{equation}

\begin{lemma}[scaled Riemann sum]
Let $F$ be continuous on $(0,\infty)$, and suppose that there exist
$\gamma_0>-1$ and $\gamma_\infty>1$ such that
\begin{equation}
    |F(x)|\le Cx^{\gamma_0},\qquad 0<x\le1,
\end{equation}
and
\begin{equation}
    |F(x)|\le Cx^{-\gamma_\infty},\qquad x\ge1.
\end{equation}
Then, as $h\to0+$,
\begin{equation}
    h\sum_{n=1}^{\infty}F(nh)
    \to
    \int_0^\infty F(x)\,dx .
    \label{eq:scaled-riemann-sum}
\end{equation}
\end{lemma}

\begin{proof}
Fix $0<\delta<R<\infty$.  On $[\delta,R]$ this is the usual convergence of
Riemann sums.  The tails are uniformly small.  If $nh<\delta$, then
\begin{equation}
    h\sum_{nh<\delta}|F(nh)|
    \le
    Ch^{1+\gamma_0}
    \sum_{n<\delta/h}n^{\gamma_0}
    \le
    C\delta^{1+\gamma_0}.
\end{equation}
This tends to zero as $\delta\to0+$ because $\gamma_0>-1$.  Similarly,
\begin{equation}
    h\sum_{nh>R}|F(nh)|
    \le
    Ch^{1-\gamma_\infty}
    \sum_{n>R/h}n^{-\gamma_\infty}
    \le
    CR^{1-\gamma_\infty},
\end{equation}
which tends to zero as $R\to\infty$.  Taking the limit first on
$[\delta,R]$ and then letting $\delta\to0+$ and $R\to\infty$ proves the
claim.
\end{proof}

\begin{proposition}
Assume $p>0$ and $0<\alpha<1$.  Then, as $t\to\infty$,
\begin{equation}
    \sigma_s(t)
    \sim
    \frac{t^{1-\alpha}}{p}
    \int_0^\infty\tau^{\alpha-2}\sin\tau\,d\tau,
    \label{eq:sigma-s-asymptotics}
\end{equation}
and
\begin{equation}
    \sigma_c(t)
    \sim
    \frac{t^{1-\alpha}}{p}
    \int_0^\infty
    \tau^{\alpha-2}\sin^2\left(\frac{\tau}{2}\right)d\tau .
    \label{eq:sigma-c-asymptotics}
\end{equation}
\end{proposition}

\begin{proof}
Put $h=t^{-1/p}$ and
\begin{equation}
    F_s(x)=x^{p-k}\sin(x^{-p}),
    \qquad
    F_c(x)=x^{p-k}\sin^2\left(\frac{x^{-p}}{2}\right).
\end{equation}
Since $k=1+\alpha p$,
\begin{equation}
    p-k=p(1-\alpha)-1.
\end{equation}
As $x\to0+$,
\begin{equation}
    |F_s(x)|+|F_c(x)|\le Cx^{p(1-\alpha)-1},
\end{equation}
and this exponent is larger than $-1$.  As $x\to\infty$,
\begin{equation}
    F_s(x)=O(x^{-k}),
    \qquad
    F_c(x)=O(x^{-p-k}),
\end{equation}
and both tails are integrable because $k>1$.  Hence the scaled Riemann sum
lemma applies to $F_s$ and $F_c$.

For $\sigma_s$ we get
\begin{equation}
    t^{-(1-\alpha)}\sigma_s(t)
    =
    h\sum_{n=1}^{\infty}F_s(nh)
    \to
    \int_0^\infty F_s(x)\,dx .
\end{equation}
With $\tau=x^{-p}$,
\begin{equation}
    \int_0^\infty x^{p-k}\sin(x^{-p})\,dx
    =
    \frac1p
    \int_0^\infty\tau^{\alpha-2}\sin\tau\,d\tau.
\end{equation}
This proves \eqref{eq:sigma-s-asymptotics}.  The same argument for $F_c$ gives
\eqref{eq:sigma-c-asymptotics}.
\end{proof}

Using \eqref{eq:single-mode-response},
\begin{equation}
    \cos\left(\left(1-\frac{1}{n^p}\right)t\right)-\cos t
    =
    \cos t\left(\cos\frac{t}{n^p}-1\right)
    +
    \sin t\sin\frac{t}{n^p}.
\end{equation}
Moreover,
\begin{equation}
    \frac{n^{-k}}{1-(1-n^{-p})^2}
    =
    \frac12 n^{p-k}\left(1+O(n^{-p})\right).
\end{equation}
Thus the leading part of the model solution is
\begin{equation}
    u(t)
    =
    \frac12\sigma_s(t)\sin t
    -
    \sigma_c(t)\cos t
    +
    o(t^{1-\alpha}).
    \label{eq:model-main-part}
\end{equation}

\begin{theorem}
Let $p>0$ and $0<\alpha=(k-1)/p<1$.  The solution of
\eqref{eq:model-oscillator}--\eqref{eq:model-forcing} with zero initial
conditions has the asymptotics
\begin{equation}
    u(t)
    =
    t^{1-\alpha}\left(C_s\sin t+C_c\cos t\right)
    +
    o(t^{1-\alpha}),
    \qquad t\to\infty,
    \label{eq:model-solution-asymptotics}
\end{equation}
where
\begin{equation}
    C_s=
    \frac{1}{2p}
    \int_0^\infty\tau^{\alpha-2}\sin\tau\,d\tau,
\end{equation}
and
\begin{equation}
    C_c=
    -\frac{1}{p}
    \int_0^\infty
    \tau^{\alpha-2}\sin^2\left(\frac{\tau}{2}\right)d\tau .
\end{equation}
Equivalently,
\begin{equation}
    u(t)
    =
    A_\alpha t^{1-\alpha}\sin(t+\phi_\alpha)
    +
    o(t^{1-\alpha}),
\end{equation}
where
\begin{equation}
    A_\alpha=\sqrt{C_s^2+C_c^2},
    \qquad
    \phi_\alpha=\arctan\frac{C_c}{C_s}.
\end{equation}
\end{theorem}

\begin{proof}
Substitute \eqref{eq:sigma-s-asymptotics} and
\eqref{eq:sigma-c-asymptotics} into \eqref{eq:model-main-part}.  The remainder
coming from the factor $O(n^{-p})$ in the denominator has an additional power
of $n^{-p}$ and is estimated by the same scaled-sum argument as
$o(t^{1-\alpha})$.
\end{proof}

\section{Subresonant Growth Regimes}

The condition $0<\alpha<1$ separates the subresonant accumulation regime from
the regime in which the coefficients decay too fast.  If
\begin{equation}
    0<\alpha<1,
\end{equation}
then the subresonant contribution grows as
\begin{equation}
    u_{\mathrm{sub}}(t)=O(t^{1-\alpha}).
\end{equation}
As $\alpha\to0$, this approaches the linear secular growth of an exact
resonance.  As $\alpha\to1-$, the growth becomes slower.

The boundary case $\alpha=1$ is logarithmic and requires a separate analysis.
For $\alpha>1$, the small denominators are compensated by the decay of the
coefficients $B_n$, and the family does not produce a leading unbounded
contribution.

\section{Return to the Partial Differential Equation}

We now show how the oscillator mechanism enters the original partial
differential equation.  With zero initial data for the first correction,
\eqref{eq:mode-oscillator} has the Duhamel representation
\begin{equation}
    \widetilde u_1(q,t)
    =
    \sum_{\lambda\in\Lambda(q)}
    B_\lambda
    \int_0^t
    \frac{\sin(\omega(q)(t-s))}{\omega(q)}
    e^{i\Omega_\lambda s}\,ds,
    \label{eq:duhamel}
\end{equation}
where $\Lambda(q)$ is the set of nonlinear interaction indices producing the
spatial mode $q$.

For one almost resonant term, the integral contains the factor
\begin{equation}
    \frac{1-e^{-i\Delta_\lambda t}}
    {\Delta_\lambda(2\omega(q)-\Delta_\lambda)}
    e^{i\omega(q)t},
    \label{eq:small-denominator-factor}
\end{equation}
up to uniformly bounded nonresonant oscillations.  Hence a selected
subresonant family contributes
\begin{equation}
    \widetilde u_{1,\mathrm{sub}}(q,t)
    \sim
    \frac{e^{i\omega(q)t}}{2\omega(q)}
    \sum_{n=1}^{\infty}
    B_n\frac{1-e^{-i\Delta_n t}}{\Delta_n}.
    \label{eq:pde-subseries}
\end{equation}

\begin{proposition}[conditional PDE consequence]
Let $f(u)=\sum_{r\ge2}f_ru^r$ be analytic near zero, and let $u_0$ be given by
\eqref{eq:u0}.  Suppose that for a spatial mode $q$ the right-hand side of the
first-correction equation contains a sequence of interactions
\[
    \lambda_n=(r,\sigma^{(n)},l_1^{(n)},\ldots,l_r^{(n)})\in\Lambda(q),
    \qquad n=1,2,\ldots,
\]
with
\begin{equation}
    l_1^{(n)}+\cdots+l_r^{(n)}=q
\end{equation}
and detuning
\begin{equation}
    \Delta_n=
    \omega(q)-\sum_{j=1}^{r}\sigma_j^{(n)}\omega(l_j^{(n)})
    \label{eq:main-detuning}
\end{equation}
satisfying
\begin{equation}
    \Delta_n\sim c n^{-p},
    \qquad c\ne0,\qquad p>0.
    \label{eq:pde-sub-family}
\end{equation}
Assume that the corresponding coefficients satisfy
\begin{equation}
    B_n=
    f_r\prod_{j=1}^{r}a_{l_j^{(n)}}^{\sigma_j^{(n)}}
    \sim b n^{-\kappa},
    \qquad b\ne0,
    \label{eq:main-coefficients}
\end{equation}
and
\begin{equation}
    0<\alpha=\frac{\kappa-1}{p}<1.
    \label{eq:main-alpha}
\end{equation}
Assume, in addition, that all other interactions in the Duhamel representation
for the mode $q$ give bounded contributions or $o(t^{1-\alpha})$, while the
selected family admits the subresonant sum asymptotics
\eqref{eq:subseries-asymptotics}.  Then the first correction contains the
component
\begin{equation}
    \widetilde u_{1,\mathrm{sub}}(q,t)
    =
    C_q t^{1-\alpha}e^{i\omega(q)t}
    +
    o(t^{1-\alpha}),
    \qquad t\to\infty,
    \label{eq:pde-mode-growth}
\end{equation}
where
\begin{equation}
    C_q=
    \frac{1}{2\omega(q)}
    \frac{b}{c}
    \int_0^\infty
    x^{p-\kappa}\left(1-\exp(-icx^{-p})\right)\,dx .
    \label{eq:Cq-constant}
\end{equation}
Consequently, in the formal expansion for \eqref{eq:pde} there is a term
\begin{equation}
    \varepsilon C_q t^{1-\alpha}e^{i(qx+\omega(q)t)}
    +\mathrm{c.c.},
\end{equation}
and the first-order expansion is asymptotic on the subresonant time layer
\begin{equation}
    \varepsilon t^{1-\alpha}\ll1,
    \qquad
    t\ll\varepsilon^{-1/(1-\alpha)}.
    \label{eq:main-time-layer}
\end{equation}
\end{proposition}

\section{Conditional Remainder Estimate}

We formulate a sufficient condition under which the formal subresonant
correction is indeed the first correction term.  This does not replace the
full analysis of all nonlinear interactions; rather, it separates the main
technical estimate required for such an analysis.

Let $\rho\ge0$, $s\ge0$, and let $\mathcal W_{\rho,s}$ be the space of Fourier
series
\begin{equation}
    v(x)=\sum_{k\in\mathbb Z}v_ke^{ikx}
\end{equation}
with norm
\begin{equation}
    \|v\|_{\rho,s}
    =
    \sum_{k\in\mathbb Z}
    |v_k|(1+|k|)^s e^{\rho |k|}.
    \label{eq:wiener-norm}
\end{equation}
This is a Banach algebra with respect to multiplication.  Therefore an
analytic $f$ defines a locally Lipschitz map in each ball of
$\mathcal W_{\rho,s}$.

Let $\mathcal D$ be the Duhamel operator for the linear equation:
\begin{equation}
    (\mathcal DG)_k(t)
    =
    \int_0^t
    \frac{\sin(\omega(k)(t-\tau))}{\omega(k)}
    G_k(\tau)\,d\tau .
    \label{eq:duhamel-operator}
\end{equation}
We use the following subresonant estimate.  For the considered set of
interactions, assume that there are constants $C>0$ and $T_0>0$, independent
of $\varepsilon$, such that for all $0<T<T_0$ in the subresonant scale,
\begin{equation}
    \|\mathcal D(G)\|_{T,\rho,s}
    \le
    C(1+T)^{1-\alpha}\|G\|_{T,\rho,s},
    \qquad
    \|G\|_{T,\rho,s}=\sup_{0\le t\le T}\|G(t)\|_{\rho,s}.
    \label{eq:subresonant-duhamel-estimate}
\end{equation}
Unlike the crude $O(T)$ estimate, \eqref{eq:subresonant-duhamel-estimate}
uses the absence of exact resonances and the presence of only subresonant
accumulation of order $T^{1-\alpha}$.  For the full nonlinear problem this
estimate must be verified separately, since one selected subresonant chain
does not exclude contributions from other resonant or almost resonant
families.

\begin{proposition}[sufficient condition for the remainder]
Assume $\omega(k)\ge\omega_0>0$, $u_0(t)\in\mathcal W_{\rho,s}$ uniformly in
$t\ge0$, and $f$ is analytic near a ball containing the range of $u_0$.  Let
$u_1$ solve
\begin{equation}
    \partial_t^2u_1+A(D_x)u_1=f(u_0),
    \qquad
    u_1|_{t=0}=\partial_tu_1|_{t=0}=0,
    \label{eq:u1-justification}
\end{equation}
and assume
\begin{equation}
    \|u_1(t)\|_{\rho,s}\le C_1(1+t)^{1-\alpha},
    \qquad 0<\alpha<1.
    \label{eq:u1-growth-bound}
\end{equation}
Assume also that \eqref{eq:subresonant-duhamel-estimate} holds for the
functions arising in the linearization of the nonlinear remainder near $u_0$.
Then there are $\delta>0$ and $\varepsilon_0>0$ such that for
$0<\varepsilon<\varepsilon_0$ and
\begin{equation}
    \varepsilon(1+T_\varepsilon)^{1-\alpha}\le\delta,
    \label{eq:justification-layer-condition}
\end{equation}
the exact solution of \eqref{eq:pde} exists on $0\le t\le T_\varepsilon$ and
has the representation
\begin{equation}
    u(x,t,\varepsilon)
    =
    u_0(x,t)+\varepsilon u_1(x,t)+R(x,t,\varepsilon),
    \label{eq:justified-expansion}
\end{equation}
with
\begin{equation}
    \sup_{0\le t\le T_\varepsilon}
    \|R(t,\varepsilon)\|_{\rho,s}
    \le
    C_2\varepsilon^2(1+T_\varepsilon)^{2(1-\alpha)}.
    \label{eq:remainder-estimate}
\end{equation}
In particular, if
\begin{equation}
    \varepsilon(1+T_\varepsilon)^{1-\alpha}\to0,
    \qquad \varepsilon\to0,
\end{equation}
then
\begin{equation}
    R=o\!\left(\varepsilon T_\varepsilon^{1-\alpha}\right)
\end{equation}
uniformly for $0\le t\le T_\varepsilon$.  Thus, in this conditional regime,
the subresonant correction is justified on the layer
\begin{equation}
    T_\varepsilon=o\!\left(\varepsilon^{-1/(1-\alpha)}\right).
\end{equation}
\end{proposition}

\begin{proof}
Put $u=u_0+v$.  Then
\begin{equation}
    v=\varepsilon\mathcal D f(u_0+v).
    \label{eq:v-fixed-point}
\end{equation}
The first approximation is $\varepsilon u_1$, where
$u_1=\mathcal D f(u_0)$.  Consider the ball
\begin{equation}
    \|v\|_{T,\rho,s}
    \le
    M\varepsilon(1+T)^{1-\alpha}.
\end{equation}
The Lipschitz property of $f$ in the Banach algebra $\mathcal W_{\rho,s}$ and
\eqref{eq:subresonant-duhamel-estimate} give
\begin{equation}
    \|\varepsilon\mathcal D(f(u_0+v)-f(u_0+\widetilde v))\|_{T,\rho,s}
    \le
    C\varepsilon(1+T)^{1-\alpha}
    \|v-\widetilde v\|_{T,\rho,s}.
\end{equation}
Under \eqref{eq:justification-layer-condition} this is a contraction.  Hence
\eqref{eq:v-fixed-point} has a unique solution in the ball.

Subtracting the first approximation,
\begin{equation}
    R=v-\varepsilon u_1
    =
    \varepsilon\mathcal D\left(f(u_0+v)-f(u_0)\right).
\end{equation}
Applying the same Lipschitz estimate and the bound for $v$ yields
\begin{equation}
    \|R\|_{T,\rho,s}
    \le
    C\varepsilon(1+T)^{1-\alpha}\|v\|_{T,\rho,s}
    \le
    C_2\varepsilon^2(1+T)^{2(1-\alpha)}.
\end{equation}
This proves \eqref{eq:remainder-estimate}.
\end{proof}

\section{Verification of the Subresonant Estimate}

We now verify the scalar subresonant estimate for the selected family.  This
does not replace the control of all other interactions in the full nonlinear
equation.  Consider
\begin{equation}
    S(t)=
    \sum_{n=1}^{\infty}
    B_n
    \frac{1-\exp(-i\Delta_n t)}{\Delta_n},
    \label{eq:abstract-subresonant-series}
\end{equation}
where
\begin{equation}
    \Delta_n=c n^{-p}+O(n^{-p-\eta}),
    \qquad
    B_n=b n^{-\kappa}+O(n^{-\kappa-\eta})
    \label{eq:abstract-subresonant-assumptions}
\end{equation}
for some $c\ne0$, $b\ne0$, $p>0$, and $\eta>0$.  Let
\begin{equation}
    \alpha=\frac{\kappa-1}{p}.
\end{equation}

\begin{proposition}
If $0<\alpha<1$, then there exists $C>0$ such that
\begin{equation}
    |S(t)|\le C(1+t)^{1-\alpha},
    \qquad t\ge0.
    \label{eq:subseries-growth-estimate}
\end{equation}
Moreover,
\begin{equation}
    S(t)
    =
    C_{\mathrm{sub}}t^{1-\alpha}
    +
    o(t^{1-\alpha}),
    \qquad t\to\infty,
    \label{eq:subseries-asymptotics}
\end{equation}
where
\begin{equation}
    C_{\mathrm{sub}}
    =
    \frac{b}{c}
    \int_0^\infty
    x^{p-\kappa}
    \left(1-\exp(-icx^{-p})\right)\,dx .
    \label{eq:subseries-constant}
\end{equation}
\end{proposition}

\begin{proof}
First, for all sufficiently large $n$,
\begin{equation}
    |\Delta_n|\ge c_0n^{-p},
    \qquad
    |B_n|\le C_0n^{-\kappa}.
\end{equation}
Also
\begin{equation}
    \left|
    \frac{1-\exp(-i\Delta_nt)}{\Delta_n}
    \right|
    \le C\min\{t,n^p\}.
    \label{eq:min-estimate}
\end{equation}
Let $N=[(1+t)^{1/p}]$.  Then
\begin{align}
    |S(t)|
    &\le
    C\sum_{n\le N}n^{-\kappa}n^p
    +
    Ct\sum_{n>N}n^{-\kappa}
    +
    C
    \notag\\
    &\le
    CN^{p-\kappa+1}
    +
    CtN^{1-\kappa}
    +
    C .
\end{align}
Since $\kappa=1+\alpha p$, both terms are
$O((1+t)^{1-\alpha})$.  This proves
\eqref{eq:subseries-growth-estimate}.

For the leading term, put $h=t^{-1/p}$ and
\begin{equation}
    F(x)=x^{p-\kappa}\left(1-\exp(-icx^{-p})\right).
\end{equation}
As $x\to0+$, $|F(x)|\le Cx^{p-\kappa}$, and
$p-\kappa=p(1-\alpha)-1>-1$.  As $x\to\infty$,
\begin{equation}
    F(x)=icx^{-\kappa}+O(x^{-\kappa-p}),
\end{equation}
which is integrable since $\kappa>1$.  Hence the scaled Riemann sum lemma
applies.  For the principal part,
\begin{equation}
    t^{-(1-\alpha)}
    \frac{b}{c}
    \sum_{n=1}^{\infty}
    n^{p-\kappa}
    \left(1-\exp(-ictn^{-p})\right)
    =
    \frac{b}{c}h\sum_{n=1}^{\infty}F(nh)
    \to
    \frac{b}{c}\int_0^\infty F(x)\,dx .
\end{equation}
The replacement of $B_n$ and $\Delta_n$ by their leading terms is justified on
the same decomposition into $n/t^{1/p}<\delta$,
$\delta\le n/t^{1/p}\le R$, and $n/t^{1/p}>R$.  On compact subintervals the
asymptotics in \eqref{eq:abstract-subresonant-assumptions} are uniform; the
tails are controlled by \eqref{eq:min-estimate}.  Thus the error is
$o(t^{1-\alpha})$, and \eqref{eq:subseries-asymptotics} follows.
\end{proof}

For the quartic Klein--Gordon family above,
\begin{equation}
    \Delta_n=
    \omega(n)+\omega(n+1)-\omega(2n+1)
    =
    \frac{3}{4n}+O(n^{-2}).
\end{equation}
If \eqref{eq:fourier-decay} holds, then $B_n\sim bn^{-3s}$.  Hence
\begin{equation}
    \alpha=3s-1.
\end{equation}
In the regime
\begin{equation}
    0<3s-1<1
\end{equation}
we obtain
\begin{equation}
    \left|
    \sum_{n=1}^{\infty}
    B_n\frac{1-\exp(-i\Delta_nt)}{\Delta_n}
    \right|
    \le
    C(1+t)^{2-3s}.
    \label{eq:kg-subresonant-estimate}
\end{equation}
If the domination condition for the remaining interactions is additionally
verified, the conditional remainder estimate gives the layer
\begin{equation}
    \varepsilon(1+T_\varepsilon)^{2-3s}\ll1.
\end{equation}

In this conditional regime the formal solution has the structure
\begin{equation}
    u(x,t,\varepsilon)
    \sim
    u_0(x,t)+\varepsilon u_1(x,t)+O(\varepsilon^2),
\end{equation}
where
\begin{equation}
    u_{1,\mathrm{sub}}(x,t)
    \sim
    \sum_{q\in\mathcal Q_{\mathrm{sub}}}
    C_qt^{1-\alpha_q}e^{i(qx+\omega(q)t)}
    +
    \mathrm{c.c.}
    \label{eq:pde-final-sub}
\end{equation}
Thus the subresonance does not create the exact secular growth $t$ of a true
resonance, but it creates a slower secular amplification,
\begin{equation}
    u_{1,\mathrm{sub}}=O(t^{1-\alpha_q}),
    \qquad 0<\alpha_q<1.
\end{equation}

\section{Numerical Illustration}

The numerical part checks the mechanism responsible for the new time scale,
rather than attempting a full simulation of the nonlinear PDE.  We give two
tests: the power law in the model oscillator and the asymptotics of the small
detuning in the quartic Klein--Gordon family.

For the model oscillator take
\begin{equation}
    p=3,\qquad k=2,\qquad \alpha=\frac13.
\end{equation}
The predicted growth exponent is $1-\alpha=2/3$.  Define
\begin{equation}
    A_s^{(N)}(t)
    =
    \frac12\sum_{n=1}^{N}n^{p-k}\sin\frac{t}{n^p},
    \qquad
    A_c^{(N)}(t)
    =
    -
    \sum_{n=1}^{N}n^{p-k}\sin^2\frac{t}{2n^p}.
\end{equation}
These are the coefficients of $\sin t$ and $\cos t$ in
\eqref{eq:model-main-part}.  For these parameters,
\begin{equation}
    C_s=
    \frac16
    \int_0^\infty\tau^{-5/3}\sin\tau\,d\tau
    \approx0.5800,
\end{equation}
and
\begin{equation}
    C_c=
    -\frac13
    \int_0^\infty\tau^{-5/3}\sin^2\frac{\tau}{2}\,d\tau
    \approx-0.3349.
\end{equation}

\begin{figure}[t]
    \centering
    \includegraphics[width=0.82\textwidth]{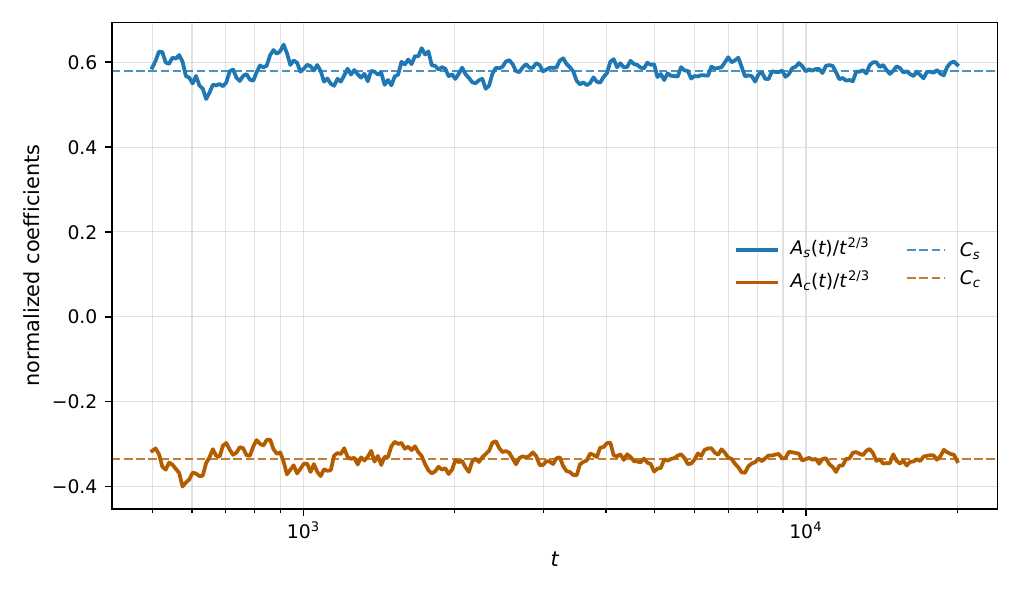}
    \caption{Normalized coefficients $A_s^{(N)}(t)/t^{2/3}$ and
    $A_c^{(N)}(t)/t^{2/3}$ for the model oscillator with $p=3$, $k=2$.
    The dashed horizontal lines are the asymptotic constants $C_s$ and $C_c$.}
    \label{fig:numeric-model-coefficients}
\end{figure}

Figure~\ref{fig:numeric-model-coefficients} uses the truncation $N=120000$.
After normalization by $t^{2/3}$ the coefficients approach constants, in
agreement with the asymptotic formulas.

For the quartic Klein--Gordon family,
\begin{equation}
    \Delta_n
    =
    \sqrt{n^2+1}
    +
    \sqrt{(n+1)^2+1}
    -
    \sqrt{(2n+1)^2+1}.
\end{equation}
The asymptotics \eqref{eq:quartic-detuning} imply $n\Delta_n\to3/4$.

\begin{figure}[t]
    \centering
    \includegraphics[width=0.88\textwidth]{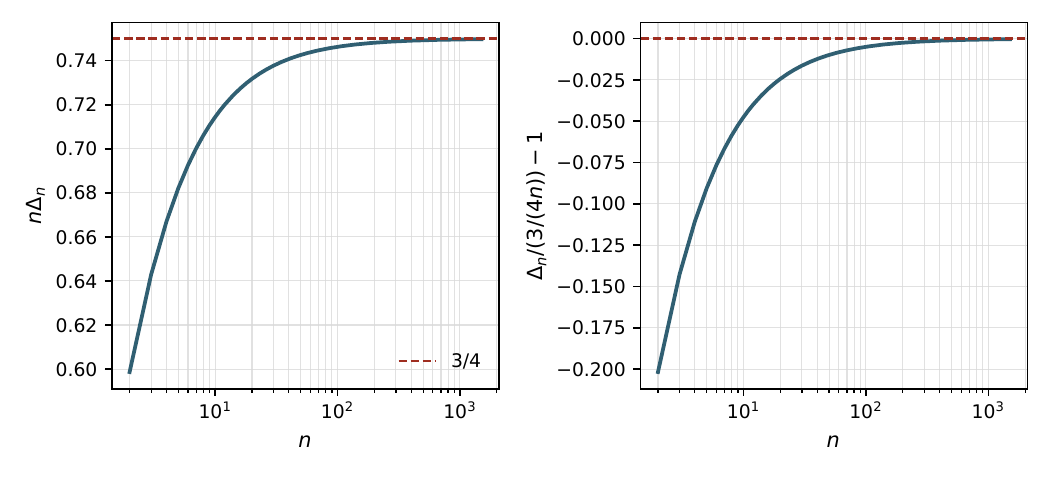}
    \caption{Numerical check of the quartic detuning.  Left: convergence of
    $n\Delta_n$ to $3/4$.  Right: relative error of the approximation
    $\Delta_n\sim3/(4n)$.}
    \label{fig:numeric-kg-detuning}
\end{figure}

Figure~\ref{fig:numeric-kg-detuning} confirms the order of the small
denominator used in \eqref{eq:kg-subresonant-estimate}.

\section{Conclusion}

We have identified a small-denominator mechanism in weakly nonlinear
dispersive dynamics.  Infinite families of nonresonant Fourier interactions
may have detunings tending to zero.  Their accumulated contribution is not
necessarily bounded; under the asymptotic assumptions
$\Delta_n\sim cn^{-p}$ and $B_n\sim bn^{-\kappa}$ it grows as
$t^{1-\alpha}$, $\alpha=(\kappa-1)/p$.  This subresonant growth lies between
bounded nonresonant dynamics and exact secular resonance.

The mechanism was proved for a model oscillator and for an abstract Duhamel
sum.  For the Klein--Gordon dispersion law we exhibited an explicit quartic
family with detuning of order $n^{-1}$.  The return to the full nonlinear PDE
was formulated as a conditional consequence, because a complete justification
requires separate control of all remaining resonant and almost resonant
interactions.  When this control is available, subresonant accumulation changes
the long-time validity scale of the perturbation expansion.


\begin{thebibliography}{99}

\bibitem{Lindstedt1883}
A. Lindstedt,
Beitrag zur Integration der Differentialgleichungen der St\"orungstheorie,
\emph{M\'emoires de l'Acad\'emie Imp\'eriale des Sciences de St. P\'etersbourg},
31(4), 1883.

\bibitem{Poincare1899}
H. Poincar\'e,
\emph{Les m\'ethodes nouvelles de la m\'ecanique c\'eleste},
Vols. 1--3, Gauthier--Villars, Paris, 1892--1899.

\bibitem{Arnold1963}
V.I. Arnold,
Small denominators and problems of stability of motion in classical and
celestial mechanics,
\emph{Russian Mathematical Surveys}, 18(6), 85--191, 1963.
doi:10.1070/RM1963v018n06ABEH001143.

\bibitem{Nayfeh1973}
A.H. Nayfeh,
\emph{Perturbation Methods},
Wiley, New York, 1973.

\bibitem{KevorkianCole1996}
J. Kevorkian and J.D. Cole,
\emph{Multiple Scale and Singular Perturbation Methods},
Springer, New York, 1996.

\bibitem{Whitham1974}
G.B. Whitham,
\emph{Linear and Nonlinear Waves},
Wiley-Interscience, New York, 1974.

\bibitem{Kuksin1993}
S.B. Kuksin,
\emph{Nearly Integrable Infinite-Dimensional Hamiltonian Systems},
Lecture Notes in Mathematics, Vol. 1556,
Springer, Berlin, 1993.

\bibitem{CraigWayne1993}
W. Craig and C.E. Wayne,
Newton's method and periodic solutions of nonlinear wave equations,
\emph{Communications on Pure and Applied Mathematics}, 46(11),
1409--1498, 1993.
doi:10.1002/cpa.3160461102.

\bibitem{Shatah1985}
J. Shatah,
Normal forms and quadratic nonlinear Klein--Gordon equations,
\emph{Communications on Pure and Applied Mathematics}, 38(5),
685--696, 1985.
doi:10.1002/cpa.3160380516.

\bibitem{BambusiGrebert2006}
D. Bambusi and B. Grebert,
Birkhoff normal form for partial differential equations with tame modulus,
\emph{Duke Mathematical Journal}, 135(3), 507--567, 2006.
doi:10.1215/S0012-7094-06-13534-2.

\bibitem{DelortSzeftel2006}
J.-M. Delort and J. Szeftel,
Long-time existence for semi-linear Klein--Gordon equations with small Cauchy
data on Zoll manifolds,
\emph{American Journal of Mathematics}, 128(5), 1187--1218, 2006.
doi:10.1353/ajm.2006.0038.

\bibitem{BambusiDelortGrebertSzeftel2007}
D. Bambusi, J.-M. Delort, B. Grebert and J. Szeftel,
Almost global existence for Hamiltonian semilinear Klein--Gordon equations
with small Cauchy data on Zoll manifolds,
\emph{Communications on Pure and Applied Mathematics}, 60(11),
1665--1690, 2007.
doi:10.1002/cpa.20181.

\bibitem{Kiselev1987}
O.M. Kiselev,
Asymptotics of the solution of the Cauchy problem for the perturbed
Klein--Gordon--Fock equation,
\emph{Zap. Nauchn. Semin. LOMI}, 165, 115--121, 1987.

\bibitem{AstafyevaKiselev2021}
P.Yu. Astafyeva and O.M. Kiselev,
Subresonant solutions of the linear oscillator equation,
in \emph{2021 International Conference ``Nonlinearity, Information and
Robotics'' (NIR)}, IEEE, 1--4, 2021.
doi:10.1109/NIR52917.2021.9666062.

\bibitem{AstafyevaKiselev2022}
P.Yu. Astafyeva and O.M. Kiselev,
Formal asymptotics of parametric subresonance,
\emph{Nelineinaya Dinamika}, 18(5), 927--937, 2022.
doi:10.20537/nd221220.

\end{thebibliography}
\end{document}